\documentclass{amsart}

\usepackage{amssymb}

\usepackage{array}
\usepackage{graphicx}         
\usepackage{amsmath}
\usepackage{amsfonts}
\usepackage{amssymb}
\usepackage{amsthm}
\usepackage{algorithm2e}
\usepackage[all]{xy}
\usepackage[bookmarks=true,linkcolor=red,citecolor=blue,urlcolor=green,colorlinks=true, breaklinks]{hyperref}
\usepackage[dvipsnames]{xcolor}

\theoremstyle{theorem}
\newtheorem{theorem}{Theorem}
\newtheorem{lemma}{Lemma}

\newtheorem{corollary}{Corollary}
\newtheorem{definition}{Definition}
\theoremstyle{remark}
\newtheorem{remark}{Remark}
\newtheorem{example}{Example}

\begin{document}

\title{A Euclidean Algorithm for Binary Cycles with Minimal Variance}
\author{Luca Ghezzi}
\address{R\&D, ABB S.p.A., Via dell’Industria, 18 - 20010 Vittuone (MI), Italy}
\email[Corresponding author]{luca.ghezzi@it.abb.com}

\author{Roberto Baldacci}
\address{DEI, University of Bologna, via Venezia, 52 - 47521 Cesena (FC), Italy}
\email{r.baldacci@unibo.it}

\begin{abstract}
The problem is considered of arranging symbols around a cycle, in such a way that distances between different instances of a same symbol be as uniformly distributed as possible. A sequence of moments is defined for cycles, similarly to the well-known praxis in statistics and including mean and variance. Mean is seen to be invariant under permutations of the cycle. In the case of a binary alphabet of symbols, a fast, constructive, sequencing algorithm is introduced, strongly resembling the celebrated Euclidean method for greatest common divisor computation, and the cycle returned is characterized in terms of symbol distances. A minimal variance condition is proved, and the proposed Euclidean algorithm is proved to satisfy it, thus being optimal. Applications to productive systems and information processing are briefly discussed.
\end{abstract}







\maketitle

\section{Introduction}

A cycle, or a cyclic order \cite{Huntington,Novak}, is a very intuitive structure describing the way some given collection of objects may be orderly arranged around a circle.
Said objects can be represented by a suitable alphabet of symbols, describing the prototypes of classes of similar objects. The number of instances for each symbol, or the multiplicity vector, denotes the number of elements in each class.
Since two paths are always available to connect two points around a circle, an order relation is uniquely determined up to specifying the positive direction around the cycle.
After this done, for each symbol instance in the cycle it is possible to define what other symbol instance is following next.
Moreover, it is also possible to count the distance, in steps, from each symbol instance to the next instance of the same symbol.
This induces naturally the mean and variance of the cycle.
One may now seek for cycles which are variance minimizers.

The problem, which we generally term \emph{Cyclic Sequencing Problem}, pertains to constrained combinatorial optimization and may be addressed in the general case by means of Mixed Integer Quadratic Programming (MIQP) techniques \cite{Bazaraa}, with a more than algebraically growing computational burden.
Nonetheless, the case of a binary alphabet, that is, when only two symbols are concerned, though with arbitrary multiplicities, allows for a direct, constructive, algorithmic solution, with linear overhead.
Specifically, a sequencing algorithm is here proposed, whose structure essentially coincides with the celebrated Euclidean method to compute the greatest common divisor (gcd) of two natural numbers, and which is additionally enriched of instructions to compile an admissible cycle.
Such cycle may be completely characterized in terms of symbol distances.
A famous analysis from Lam\'{e} \cite{Lame} proves that the Euclidean algorithm requires at most a number of steps which is five times the number $h$ of base-10 digits of the smaller number in the couple whose gcd is sought for \cite{Grossman,Honsberger}. Consequently, also the proposed Euclidean Sequencing Algorithm is $O(h)$ in the worst case.

The analysis of the involved mathematical programs allows deducing necessary and sufficient conditions for variance minimality.
It is seen that the proposed Euclidean sequencing algorithm satisfies these conditions and it is therefore optimal.
The reason behind the successfulness of the algorithm is readily found in the strong penalization operated by squares (or, more generally, by higher than linear powers) over deviations of distances from their average value.
Additionally, a remarkably simple argument, ultimately connected to the invariance of the number of steps required to complete a round trip around the cycle, shows that the mean is invariant to cycle permutations.
This two facts, mean invariance and square penalization, easily allow to prove the theory here developed.

Euclid's original algorithm was introduced around 300 B.C. in the celebrated Elements \cite{Euclid}, Book VII, Proposition 2, to find the greatest common measure of two given numbers not relatively prime (according to the geometrically inspired terminology of the times). Also due to the geometrical identification of numbers and segment lengths, the exact presentation of the algorithm resorts on repeated subtractions in place of divisions and it thus formally differs from, but is equivalent to the vest currently adopted.

Some few applications to even and cyclically repeated distribution of symbols are found in the literature. Toussaint \cite{Toussaint} connects the problem with music and shows that traditional musical rhythms are generated by the Euclidean algorithm (with a binary alphabet, if expressed in the terminology of the present work) and are therefore dubbed \emph{Euclidean rhythms}. The approach is algorithmic and neither mathematical formalism to express an evenness metric, nor proofs of optimality are there found.

The same problem has been analyzed by Bjorklund \cite{Bjorklund2,Bjorklund1} in connection with spallation neutral source (SNS) accelerators as used in nuclear physics.
Translated in the terminology of the present work, the problem is there to evenly sequence symbols from a binary alphabet, where one symbol is interpreted as 1 and the other as 0.
Despite the natural approach with Boolean symbols, the nullity of 0 seems to have induced the Author to choose to only consider the 1's, neglecting the 0's when defining suitable evenness metrics.
On the other hand, variance as a metric for evenness, as defined in the present work, is computed with reference to all symbols in the alphabet, as opposed to Bjorklund's approach, which leads to an unsatisfactory metric.
It is easily seen that the former approach (but not the latter, as correctly pointed out in \cite{Bjorklund2}) manages to fulfill Bjorklund's very natural requirements for a good evenness metrics (viz., invariance under rotation, efficient computation, null value for perfectly even distributions), with the additional advantage of simplicity; see section \S\ref{sec:apps} for examples.

A mathematical analysis of the problem at hand is found in Demaine et Al. \cite{Demaine}, with remarkable theoretical results applying to the product of the Euclidean sequencing algorithm.
Also in this case, evenness metrics are computed with reference to one sole symbol 1 in a binary alphabet.
We argue that this be a bias induced by the underlying applicative contexts, where 1 stands for a physical event (a pulse, or a musical note), while 0 stands for nothingness, just waiting for an event to happen.
Contrarily, the inspiring applicative problem for the present work resides in industrial manufacturing systems, where symbols stand for different product types, that is, real objects, which cannot be associated with nothingness.
Moreover, the formal symmetry in considering all symbols in the alphabet, when defining metrics, is also for mathematical beauty, as well as extending immediately to more than binary alphabets.
Demaine et Al. approach the problem of evenness in terms of maximization of a metric constituted by the sum of chordal distances between 1 symbols.

Relationships with so-called Euclidean strings (not relevant to the present work) are also discussed in the references above; see, e.g., Ellis et Al. \cite{Ellis}.


\subsection{Contributions of this paper}

In this paper, we introduce a new problem, the \emph{Cyclic Sequencing Problem} (CSP), motivated by a real application from industrial manufacturing systems. Our distinct contributions in this paper are as follows:
\begin{itemize}
  \item We propose a novel mathematical programming formulation for the CSP and a relaxation that is used to derive valid lower bounds.
  \item For the special case of binary cycles, we propose an algorithm, the Euclidean Sequencing Algorithm (ESA), that is similar to the algorithm proposed by Demaine et al. \cite{Demaine}. In contrast, the present analysis approaches the problem of evenness in terms of minimization of a suitably defined variance in the distribution of all symbols in the alphabet. One may immediately appreciate the difference between the former approach, which is metrical and, as such, set in Euclidean spaces (as a special case of Hilbert spaces), and the latter approach, which is essentially combinatorial. In other terms, the metrical properties of a circle, including chordal distances, are not essential to the present analysis. Since the Euclidean algorithm in the case of a binary alphabet returns a solution which is optimal for the metrics defined in \cite{Demaine} as well as the one here defined, a further connection between all of them is thus encountered.
  \item Also for the special case of binary cycles, we prove a minimal variance condition and we show that the proposed ESA satisfies it, thus being optimal.
\end{itemize}

The outline of the paper is as follows.
After introducing the standard notation adopted, including the concept of a cycle, in Section \ref{sec:notation} (raw) moments are defined for cycles, along with the problem of sequencing around a cycle a set of symbols from a given alphabet and with prescribed multiplicities. The mean is then introduced as the first moment, and its invariance under permutation is proved, in Section \ref{sec:roundtrip}.
Central moments, and particularly variance, are defined in Section \ref{sec:minvariance}, where the variance minimization problem is also introduced and shown to be equivalent to the minimization of the second (raw) moment.
The Euclidean sequencing Algorithm (ESA) is introduced and exemplified in Section \ref{sec:algorithm}, where the cycle returned is completely characterized in terms of the distances for the more abundant and the less abundant symbol in the binary alphabet.
The optimality condition is introduced and proved in section \ref{sec:optimality}, together with the crucial result that the ESA satisfies said condition and is optimal for the problem at hand.
Finally, applications are discussed in section \ref{sec:apps}.

\section{Notation and definitions}
\label{sec:notation}

Following the standard notation,
$\mathbb{N}:=\{1,2,\ldots\}$ is the set of natural numbers (0 is excluded),
$\mathbb{N}_0:=\mathbb{N}\cup\{0\}$, $\mathbb{Z}:=\{\ldots,-2,-1,0,1,2,\ldots\}$ is the ring of integers, $\mathbb{Q}$ is the field of rationals, and
$\mathbb{R_+}:=\{x\in\mathbb{R}\,|\,x\ge 0\}$ is the half-line of non negative reals.
Sets of symbols (order does not matter, repetitions are not accounted for) are represented by curly brackets (e.g., $\{1,2\}=\{2,1\}=\{2,1,2\}$), while sequences (order matters, repetitions are accounted for) are represented as vectors, i.e., by square brackets (e.g., $[1,2]\neq[2,1]\neq[2,1,2]$).
Let $\textrm{Sym}(N)$ denote the symmetric group on a finite set of $N$ symbols, i.e., the set of all permutations of $N$ distinct symbols, forming a (generally non Abelian) group with reference to function composition $\circ$.
For $a,b\in\mathbb{N}$, the classic algebraic notations $a\,|\,b$ stands for $a$ divides $b$, i.e., $b/a\in\mathbb{N}$, and $a \nmid b$ stands for $a$ does not divide $b$, i.e., $b/a\notin\mathbb{N}$.

Let $\mathcal{A}:=[a_1,a_2,\ldots,a_n]$ be an \emph{alphabet} of $n:=|\mathcal{A}|$ distinct \emph{symbols} and
$\mathbf{m}\in\mathbb{N}^n$ a vector of $n$ positive integers representing prescribed \emph{multiplicites} of said symbols, in such a way that
\begin{equation}
[a_k\,|\,m_k]
:=[\underbrace{a_k, \ldots, a_k}_{m_k \textrm{ times}}],\qquad a_k\in\mathcal{A},
\end{equation}
represents the $m_k$ repetitions of symbol $a_k$.
The couple $\mathcal{S}:=(\mathcal{A},\mathbf{m})$ shall be termed a \emph{cyclic sequencing problem}.
Let $N:=\|\mathbf{m}\|_1=\sum_{k=1}^nm_k$ be the total number of all symbols in the cyclic sequencing problem, accounting for repetitions.
It follows from the definitions that all symbols in the alphabet are used at least once: the case of some null multiplicity may always be reduced to a smaller alphabet comprising only those symbols with positive multiplicity.

Let $\mathbb{Z}_N:=\mathbb{Z}/N\mathbb{Z}$ be the quotient group with reference to addition modulo $N$, i.e., the set of cosets $k+N\mathbb{Z}$, for $k\in\{0,1,\ldots,N-1\}$. Clearly, $\mathbb{Z}_N$ is isomorphic to the group $(\{[0]_N,[1]_N,\ldots,[N-1]_N\};+)$, which is the same as $(\{[1]_N,[2]_N,\ldots,[N]_N\};+)$.
Let $C\,:\,\mathbb{Z}_N\rightarrow\mathcal{A}$ be a mapping from the group of rest classes modulo $N$ to the alphabet, with $C_j:=C([j]_N)$ for brevity.
Any such mapping is termed a \emph{cyclic order}, i.e., informally, a way of sequencing a set of symbols around a circle. A set with a cyclic order is termed a \emph{cycle} and, with a little abuse, we shall treat the latter two terms as synonyms. Cycles are characterized by a cyclic, asymmetric, transitive, total ternary order relation $(a,b,c)$, indicating that $b$ lies after $a$ and before $c$. Among the two possible verses, we shall conventionally refer to the one induced by increasing $[j]_N$ in $\mathbb{Z}_N$, with no loss of generality. In the very interesting $n=2$ case, we shall refer to a \emph{binary alphabet} and \emph{binary cycles}.

Let $\mathcal{J}_k:=C^{-1}(a_k)=\{[j]_N\in\mathbb{Z}_N\,|\,C_j=a_k\}$ be the counter image of the $k$th symbol in the alphabet.
Let  $\Omega(\mathcal{A},\mathbf{m}):=\{C\,:\,\mathbb{Z}_N\rightarrow\mathcal{A}\;|\;m_k=|\mathcal{J}_k|,\forall k\in\{1,\ldots,n\}\}$ be the set of \emph{admissible} cycles, i.e., those agreeing with all prescribed multiplicities.
Said set is not empty, for it contains at least the \emph{base cycle}
\begin{equation}\label{eqn:base}
\hat{C}:=\bigsqcup_{k=1}^n[a_k\,|\,m_k]
=[
\underbrace{a_1,\ldots,a_1}_{m_1 \textrm{times}},
\underbrace{a_2,\ldots,a_2}_{m_2 \textrm{times}},
\ldots,
\underbrace{a_n,\ldots,a_n}_{m_n \textrm{times}}
],
\end{equation}
where $\sqcup$ stands for concatenation. The adjective \emph{base} is because any other cycle $C$ may be expressed as the action on $\hat{C}$ of some permutation, i.e., $C=\hat{C}\circ\tau$ for some $\tau\in\textrm{Sym}(N)$, according to the following  diagram, which is commutative by construction.
\begin{displaymath}
\xymatrix
{
\mathbb{Z}_N \ar[r]^{\tau} \ar[dr]_C &
\mathbb{Z}_N \ar[d]^{\hat{C}} \\
&
\mathcal{A}
}
\end{displaymath}
Notice that, in \eqref{eqn:base}, we have used the same symbol $\hat{C}$ for both the application and the result produced by the application, with a little abuse.
Owing symbol multiplicity, the number of distinct admissible cycles does not exceed $|\textrm{Sym}(N)|=N!$ and, precisely, $|\Omega(\mathcal{S})|:=N!/\prod_{k=1}^n m_k!$.
The tally may be further lowered if the theory has to developed up to the addition of some rest class $[j_0]_N$ (i.e., modulo the symmetry induced by rotations of the circle) and/or up to the verse (i.e., modulo the symmetry induced by the two possible orderings over the circle).

\begin{definition}\label{def:distance}
Let a \emph{step} be a unit move around the cycle, that is, from $[j]_N$ to $[j+1]_N$ for some $[j]_N\in\mathbb{Z}_N$. Let $\Delta_j$ be the distance, in steps, from the $j$th entry in the cycle $C$ (i.e., $C_j=a_k\in\mathcal{A}$, for some $k$) to the next instance in the cycle of the same symbol $a_k$. In case $m_k=1$, i.e., if some symbol $a_k$ appears only once, the definition applies with reference to a full round trip from and to the unique instance, with $\Delta_j=N$.
\end{definition}

In analogy with other well-known contexts, like the theory of probability density functions in statistics, let us introduce the following
definitions.
\begin{definition}[Sub-moments of a cycle]\label{def:submoment}
For a given cycle $C\in\Omega(\mathcal{A},\mathbf{m})$ and $a_k\in\mathcal{A}$, let
\begin{equation}\label{eqn:submoment}
\mathcal{M}_p(a_k):=\frac{1}{N}\sum_{j\in\mathcal{J}_k}\Delta_j^p.
\end{equation}
be its $(p,k)$th sub-moment.
\end{definition}

\begin{definition}[Raw moment of a cycle]\label{def:moment}
For a given cycle $C\in\Omega(\mathcal{A},\mathbf{m})$, let
\begin{equation}\label{eqn:moment}
\mathcal{M}_p(C)
:=\frac{1}{N}\sum_{j=1}^{N}\Delta_j^p
 =\sum_{k=1}^n\mathcal{M}_p(a_k)
\end{equation}
be its $p$th (raw) moment.
\end{definition}

\begin{remark}
The additive decomposition in terms of the $(p,k)$th sub-moments follows from linearity of summation.
Despite such decomposition being straightforward, it will prove remarkably helpful developing the following results.
\end{remark}

\section{Mean of cycle and the round trip lemma}
\label{sec:roundtrip}

\begin{definition}[Mean of a cycle]\label{def:mean} For a given cycle $C\in\Omega(\mathcal{A},\mathbf{m})$, let
\begin{equation}\label{eqn:mean}
\mu(C):=\mathcal{M}_1(C):=\frac{1}{N}\sum_{j=1}^{N}\Delta_j
\end{equation}
be its mean.
\end{definition}
The first moment is termed mean of the cycle because it measures the mean distance in steps between symbols of the same kind.
Contrarily to other contexts, the mean is an invariant for the cyclic sequencing problem, meaning that it only depends on the (cardinality $n$ of the) alphabet $\mathcal{A}$, but not on how symbols are sequenced. Moreover, the mean is immediately known \emph{a priori} and, maybe surprisingly, it is always integer, despite its definition includes a ratio.
The reason for this as simple as stunning result lies in the following

\begin{lemma}[Round trip lemma]\label{th:roundtrip}
Let $(\mathcal{A},\mathbf{m})$ be a cyclic sequencing problem. Then
\begin{equation}\label{eqn:roundtrip}
\sum_{j\in\mathcal{J}_k}\Delta_j=N,\qquad\forall k\in\{1,\ldots,n\}.
\end{equation}
Moreover, for any admissible cycle $C\in\Omega(\mathcal{A},\mathbf{m})$ the mean is invariant under symbol permutations, integer valued and equal to the number of distinct symbols, that is,
\begin{equation}\label{eqn:invariance}
\mu(C)=n\in\mathbb{N}.
\end{equation}
\end{lemma}

\begin{proof}
Result (\ref{eqn:roundtrip}) follows immediately because an exactly complete round trip over the cycle is obtained by stepping from one symbol instance to the next instance of the same symbol, and then forward to the next instance and so on, visiting exactly once any instance of the same symbol until the starting instance is reached again, with exactly $N$ steps.

Next, let $C:=\hat{C}\circ\tau$ be a generic admissible cycle, with $N$ total entries.
Let $\Delta_j$ be the distance in steps from the $j$th entry to the next occurrence of the same symbol.
After sub-moment decomposition and applying (\ref{eqn:roundtrip}) to any symbol $a_k\in\mathcal{A}$, one gets
$$
\mu(C)
=\mathcal{M}_1(C)
=\sum_{k=1}^n\mathcal{M}_1(a_k)
=\sum_{k=1}^n\frac{1}{N}\sum_{j\in\mathcal{J}_k}\Delta_j
=n\in\mathbb{N}.
$$
\end{proof}

It follows that one may write $\mu(\mathcal{S})=\mu(\mathcal{A},\mathbf{m})$ instead of $\mu(C)$ for $C\in\Omega(\mathcal{A},\mathbf{m})$, or simply $\mu$ whenever clear from the context.
Moreover, consistently with Lemma \ref{th:roundtrip}, we shall use the value $n$ for $\mu$.

\begin{example}\label{eg:0}
Let $(\mathcal{A},\mathbf{m})=([a_1,a_2],[8,4])$, meaning 8 instances of symbol $a_1$ and 4 instances of symbol $a_2$, for a total $N_1=12$ items. Considering, e.g., the  three admissible cycles
$$
\begin{array}{lcl}
\hat{C} & = & [
\overset{1}{a_1},
\overset{1}{a_1},
\overset{1}{a_1},
\overset{1}{a_1},
\overset{1}{a_1},
\overset{1}{a_1},
\overset{1}{a_1},
\overset{5}{a_1},
\overset{1}{a_2},
\overset{1}{a_2},
\overset{1}{a_2},
\overset{9}{a_2}
],
\\
C_1 & = & [
\overset{3}{a_1},
\overset{1}{a_2},
\overset{3}{a_2},
\overset{1}{a_1},
\overset{2}{a_1},
\overset{4}{a_2},
\overset{1}{a_1},
\overset{1}{a_1},
\overset{2}{a_1},
\overset{4}{a_2},
\overset{1}{a_1},
\overset{1}{a_1}
],
\\
C_2 & = & [
\overset{1}{a_1},
\overset{2}{a_1},
\overset{3}{a_2},
\overset{1}{a_1},
\overset{2}{a_1},
\overset{3}{a_2},
\overset{1}{a_1},
\overset{2}{a_1},
\overset{3}{a_2},
\overset{1}{a_1},
\overset{2}{a_1},
\overset{3}{a_2}
],
\end{array}
$$
where the distances $\Delta_j$ have been reported on top of symbols, one may easily verify that
$\sum_{j\in\mathcal{J}_1}\Delta_j=\sum_{j\in\mathcal{J}_2}\Delta_j=12$, or also $\sum_{j=1}^{12}\Delta_j=24=Nn$, so that $\mu=2=n$ in all cases.
\end{example}

\section{Variance and the minimality problem}
\label{sec:minvariance}

It is now natural to introduce the following definitions.
\begin{definition}[Central moment of a cycle] For a given cycle $C\in\Omega(\mathcal{A},\mathbf{m})$, let
\begin{equation}
\mathcal{M}^\ast_p(C):=\frac{1}{N}\sum_{j=1}^{N}(\Delta_j-\mu)^p
\end{equation}
be its central moment of order $p$.
\end{definition}

\begin{definition}[Variance of a cycle]\label{def:variance}
For a given cycle $C\in\Omega(\mathcal{A},\mathbf{m})$, let
\begin{equation}\label{eqn:variance}
\sigma^2(C):=\mathcal{M}^\ast_2(C):=\frac{1}{N}\sum_{j=1}^{N}(\Delta_j-\mu)^2
\end{equation}
be its variance.
\end{definition}
The second central moment is termed variance because it measures the dispersion around the mean of the distribution of the distances, in steps, between instances of the same symbols.
Following a standard argument, central moments may be connected to (raw) moments of the kind (\ref{eqn:moment}), for which decomposition in sub-moments holds. In the case of variance, the following simple result holds.

\begin{lemma}
For any cycle $C\in\Omega(\mathcal{A},\mathbf{m})$,
\begin{equation}
\sigma^2(C)=\mathcal{M}_2(C)-n^2.
\end{equation}
\end{lemma}

\begin{proof}
Squaring the r.h.s. in (\ref{eqn:variance}),
$$
\sigma^2(S)
=\frac{1}{N}\sum_{j=1}^{N}\Delta_j^2
-2\mu\frac{1}{N}\sum_{j=1}^{N}\Delta_j
+\mu^2,
$$
from which the result follows thanks to the definition (\ref{eqn:mean}) of mean, the latter being equal to $n$ for round trip Lemma \ref{th:roundtrip}.
\end{proof}

We now appreciate a fundamental consequence of the round trip Lemma \ref{th:roundtrip}:
due to the invariance of $\mu$, variance $\sigma^2(C)$ and the second moment $\mathcal{M}_2(C)$ only differ by an additive constant, pointless in optimization problems. This implies the following

\begin{corollary}\label{th:optidentity}
Let $(\mathcal{A},\mathbf{m})$ be a cyclic sequencing problem.
The search space being $\Omega(\mathcal{A},\mathbf{m})$, minimization of variance is equivalent to minimization of the second moment, or
\begin{equation}
\arg\min\sigma^2(C)=\arg\min\mathcal{M}_2(C),
\end{equation}
subject to any common (and possibly empty) set of constraints applied to both problems.
\end{corollary}

We are interested in cycles of minimal variance, that is, to solve
\begin{equation}\label{eqn:varminorig}
\min_{\tau\in\textrm{Sym(N)}}\sigma^2(\hat{C}\circ\tau).
\end{equation}
Problem \eqref{eqn:varminorig} can be explicitly stated  as a mixed integer quadratic program (MIQP) as follows.
Let us first consider the ordered sequence $\hat{\mathcal{J}}=[1,2,\ldots,i,\ldots,N]$, where $i$ stands for $[i]_N$, for brevity.
The base cycle $\hat{C}$ maps $\hat{\mathcal{J}}$ to \eqref{eqn:base}.
Given the repetitions of symbols of $\mathcal{A}$ according to multiplicities $\mathbf{m}$, it is convenient to partition $\hat{\mathcal{J}}$ so to put in evidence those indexes mapped by $\hat{C}$ to a same $a_k\in\mathcal{A}$, $\forall k\in\{1,\ldots,n\}$.
To this end, let $\hat{m}_k:=\sum_{h=1}^k m_{h}$, $\forall k\in\{1,\ldots,n\}$, so that $\hat{m}_n=N$.
Then, let $\check{m}_1:=1$ and $\check{m}_k:=\hat{m}_{k-1}+1,\forall k\in\{2,\ldots,n\}$.
It is apparent that $\check{m}_k$ and $\hat{m}_k$ mark, respectively, the initial and final position in $\hat{\mathcal{J}}$ of all indexes mapped by $\hat{C}$ to symbol $a_k\in\mathcal{A}$.
Consequently, $\hat{\mathcal{J}}_k:=[\check{m}_k,\check{m}_k+1,\ldots,\hat{m}_k-1,\hat{m}_k]=\hat{C}^{-1}(a_k)$, $\forall k\in\{1,\ldots,n\}$ and
$\hat{\mathcal{J}}=\hat{\mathcal{J}}_1\sqcup \hat{\mathcal{J}}_2\sqcup \ldots \sqcup \hat{\mathcal{J}}_n$.
Clearly, $|\hat{\mathcal{J}}_k|=m_k$ by construction.

Let $\mathcal{J}:=\tau(\hat{\mathcal{J}})=[i_1,i_2,i_3,\dots,i_N]$, for suitable $i_1,i_2,i_3,\dots,i_N$, be the permuted image of $\hat{\mathcal{J}}$ through $\tau$.
Without loss of generality, to reduce the number of equivalent permutations up to indistinguishable reshuffles of instances of a same symbol, we assume that $\tau(1)=1$ (i.e., $i_1=1$ is chosen as origin in the cycle) and that for any pair $r,s$ in a same $\hat{\mathcal{J}}_k$, $\tau(r)$ precedes $\tau(s)$ in $C$.
Let $x_{ij}$, $i,j \in \hat{\mathcal{J}}$, $i \neq j$, be a (0-1) binary variable equal to 1 if $j$ immediately follows $i$ in $\mathcal{J}$, 0 otherwise. In addition, let $\theta_i$, $i \in \hat{\mathcal{J}}$, be an arbitrary real number, such that if item $i\in \hat{\mathcal{J}}$ is sequenced $h$th after item 1 in $\mathcal{J}$, then $\theta_i=h=\tau(i)-1$, with $\theta_1=0=\tau(1)-1$. The term $\theta_i$ suggests the usual angular coordinate to span the circle.

The mathematical formulation of \eqref{eqn:varminorig} is as follows:
\begin{alignat}{3}
\min  & \sum_{i \in \hat{\mathcal{J}}} \Delta^2_i \label{F.Obj}\\
 s.t. & \sum_{j \in \hat{\mathcal{J}}}x_{ij}=1, & \forall i \in \hat{\mathcal{J}} \label{F.DegreeA}\\
      & \sum_{i \in \hat{\mathcal{J}}}x_{ij}=1, & \forall j \in \hat{\mathcal{J}} \label{F.DegreeB}\\
      & \theta_i - \theta_j + N x_{ij} \leq N-1, & \quad \forall i,j \in \hat{\mathcal{J}}\setminus\{1\}, i \neq j \label{F.MTZ}\\
      & \theta_i \leq \theta_{i+1}, & \forall i \in \hat{\mathcal{J}}_k \setminus \{\hat{m}_k\}, \forall k\in\{1,\ldots,n\} \label{F.uk}\\
      & \Delta_i = \theta_{i+1}-\theta_i, & \forall i \in \hat{\mathcal{J}}_k\setminus \{\hat{m}_k\}, \forall k\in\{1,\ldots,n\} \label{F.dka}\\
      & \Delta_{\hat{m}_k} = N-\theta_{\hat{m}_k}+\theta_{\check{m}_k}, & \forall k\in\{1,\ldots,n\}\label{F.dkb}\\
      & x_{ij} \in \{0,1\}, &\forall i, j  \in \hat{\mathcal{J}} \label{F.IntX}\\
      & \theta_{i} \geq 0, & \forall i \in \hat{\mathcal{J}} \label{F.ConU}\\
      & \Delta_i \geq 0, & \forall i \in \hat{\mathcal{J}}. \label{F.ConD}
\end{alignat}

The plurality of nonnegative variables $\Delta_i$ in program \eqref{F.Obj}-\eqref{F.ConD} represent the distance in $\mathcal{J}$ between item $i$ and the next item mapped to the same symbol and therefore correspond exactly with the distance $\Delta_j$ as in Definition \ref{def:distance}, with $j=\tau(i)$.
The objective function \eqref{F.Obj} states to minimize the second moment (corresponding to total variance minimization, for Corollary \ref{th:optidentity}). Constraints \eqref{F.DegreeA},\eqref{F.DegreeB} together with constraints \eqref{F.MTZ} impose a cycle in solution. Constraints \eqref{F.MTZ} also define the values of variables $\theta_i$.
Constraints \eqref{F.dka},\eqref{F.dkb} define the values of distance variables $\Delta_i$, $\forall i \in \hat{\mathcal{J}}$.

\begin{remark}
Formulating the problem in the domain of $\tau$, i.e., in $\hat{\mathcal{J}}$, instead of in the image of $\tau$, i.e., in $\mathcal{J}$, brings the benefit that all the involved summations and references are in terms of index sets known a priori than solving the problem.
\end{remark}

\section{The Euclidean Sequencing Algorithm (ESA)}
\label{sec:algorithm}

Let us consider the case of binary cycles of minimal variance.
Despite the presence of integrality constraints, this special case is surprisingly simple, to the extent that it may be directly solved algorithmically, according to a procedure resembling the celebrated one for the gcd of two integers and attributed to Euclid.
Therefore the algorithm, which holds for a cyclic sequencing problem $(\mathcal{A},\mathbf{m})$ with $n=2$,
is termed the \emph{Euclidean Sequencing Algorithm} (ESA).
Before presenting the computational machinery, we discuss the underlying rationale by means of examples.
The proof for minimal variance follows in the next section.
For the sake of brevity, we shall use the formal product formalism $abc=[a,b,c]$ and introduce formal powers of the kind $a^2bc^3=[a,a,b,c,c,c]$. Clearly, this formal product is not commutative.

\begin{example}\label{eg:1}
The cyclic sequencing problem $(\mathcal{A},\mathbf{m})=([a_1],[m_1])$, with $n=1$ and $N=m_1$, admits only the cycle $C=a_1^N$, which is trivially of minimal variance.
\end{example}

\begin{example}\label{eg:2}
Let us start with an intuitive case.
Let $(\mathcal{A},\mathbf{m})=([a_1,a_2],[8,4])$, meaning 8 instances of symbol $a_1$ and 4 instances of symbol $a_2$, for a total $N=12$ items. Here and in the sequel, let $A$ (possibly with a subscript) denote the symbol with the greatest number of repetitions (here, $a_1$) and $B$ (possibly with a subscript) the other (here $a_2$).
Intuitively, minimal variance means the regular repetition of a fixed scheme.
Since $4$ is a perfect divisor of $8$ and $8=4\cdot 2$, this is achieved by, e.g., $C=[A,A,B,A,A,B,A,A,B,A,A,B]=AABAABAABAAB$, or by other cycles obtained by applying a rotation, like, e.g., $ABAABAABAABA$.
\end{example}

\begin{figure}[t!]
\centering
\includegraphics[width=0.7\linewidth]{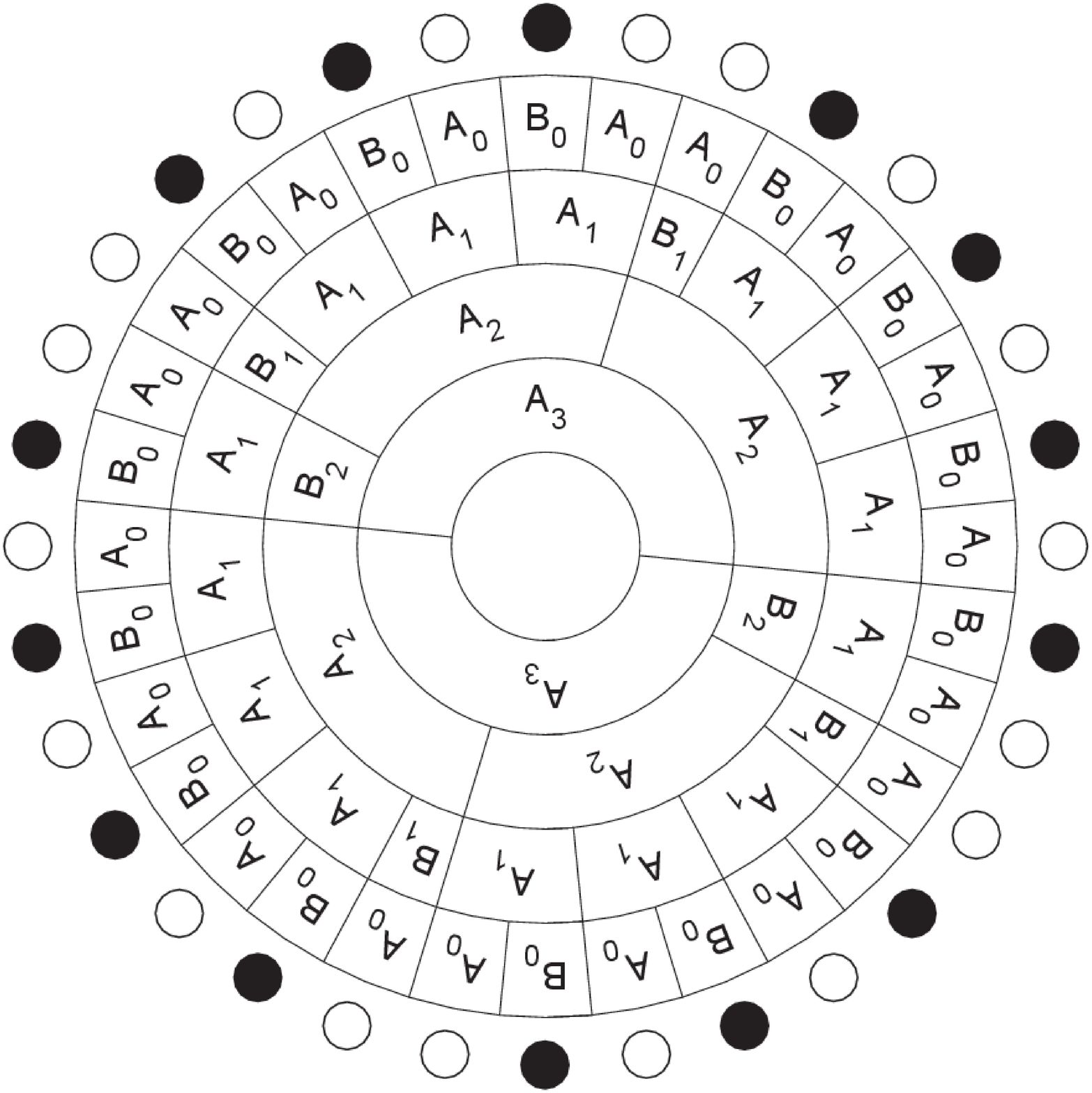}
\caption{The ESA applied to $(\mathcal{A},\mathbf{m})=([\circ,\bullet],[18,14])$.}
\label{fig:ruota}
\end{figure}

\begin{example}\label{eg:3}
We now face a case where perfect divisibility is not met; see Table \ref{tab:ESA}.
Let $(\mathcal{A}_1,\mathbf{m}_1)=([a_1,a_2],[18,14])$, meaning 18 instances of symbol $a_1$ (denoted by $A_0$) and 14 instances of symbol $a_2$ (denoted by $B_0$), for a total $N_1=32$ items.
We notice that $18=1\cdot 14+4$, where $1=\left\lfloor{18/14}\right\rfloor$ is the quotient of the integer division $18/14$ and $4\equiv 18\mod 14$ is the rest of that division (it is an elementary result that this decomposition is unique and that the rest is always less than the divisor).
We can now think of addressing the smaller cyclic sequencing problem $(\mathcal{A}_2,\mathbf{m}_2):=([A_1,B_1],[14,4])$, with 14 repetitions of the new symbol $A_1:=[A_0,B_0]=A_0B_0$, using all instances of $B_0$ but only 14 out of 18 instances of $A_0$, and the remaining 4 repetitions of $B_1:=A_0$ symbol.
By reasoning similarly, since $14=3\cdot 4+2$, one gets the smaller cyclic sequencing problem $(\mathcal{A}_3,\mathbf{m}_3):=([A_2,B_2],[4,2])$, with 4 instances of $A_2:=A_1^3B_1$ and 2 instances of $B_2:=A_1$.
We have met the perfect divisibility condition, since $4=2\cdot 2$.
Then one either closes as in Example \ref{eg:2}, or a further iteration is worked out with $(\mathcal{A}_4,\mathbf{m}_4):=([A_3],[2])$, where $A_3:=A_2^2B_2$, so to close as in Example \ref{eg:1}.
In both cases one finds the cycle $C=A_3^2$, that is, $C=(A_2^2B_2)^2$, which has minimal variance in, resp., $\Omega(\mathcal{A}_4,\mathbf{m}_4)$ or $\Omega(\mathcal{A}_3,\mathbf{m}_3)$.

Backtracking is straightforwardly accomplished by backward substitution of the symbols progressively defined, and we claim that variance minimality be preserved for every $\Omega(\mathcal{A}_i,\mathbf{m}_i)$.
A possible way to keep a log during algorithm execution is reported in Table \ref{tab:ESA}.
In the case at hand,
$$
\begin{array}{lcl}
C
& = & A_3^2 \\
& = & (A_2^2B_2)^2 \\
& = & ((A_1^3B_1)^2A_1)^2 \\
& = & (((A_0B_0)^3A_0)^2A_0B_0)^2 \\
& = & a_1a_2a_1a_2a_1a_2a_1a_1a_2a_1a_2a_1a_2a_1a_1a_2 \\
&   & a_1a_2a_1a_2a_1a_2a_1a_1a_2a_1a_2a_1a_2a_1a_1a_2. \\
\end{array}
$$
\end{example}

\begin{table}[t]
\centering
\begin{tabular}{clrrrrrll}
\hline
\\
$i$ & $\mathcal{A}_i$ & $N_i$ & $P_i$ & $D_i$ & $Q_i$ & $R_i$ & $A_i$ & $B_i$ \\
\\
\hline
\\
0 &             &    &    &    &   &   & $A_0=a_1$       & $B_0=a_2$    \\
1 & $[A_0,B_0]$ & 32 & 18 & 14 & 1 & 4 & $A_1=A_0^1B_0$  & $B_1=A_0$  \\
2 & $[A_1,B_1]$ & 18 & 14 &  4 & 3 & 2 & $A_2=A_1^3B_1$  & $B_2=A_1$  \\
3 & $[A_2,B_2]$ &  6 &  4 &  2 & 2 & 0 & $A_3=A_2^2B_2$  \\
\\
\hline
\\
  &               &    &    &    &   &   & $C=A_3^2$ \\
\\
\hline
\end{tabular}
\caption{Application of the Euclidean Sequencing Algorithm (ESA).}
\label{tab:ESA}
\end{table}

A graphical representation of the ESA applied to $(\mathcal{A},\mathbf{m})=([\circ,\bullet],[18,14])$, i.e., $a_1:=\circ$ and $a_2:=\bullet$, is shown in Figure \ref{fig:ruota}, where the stages of the iterative algorithm are shown in progressively smaller and inner annuli.
In the rationale behind this simple machinery resides the essence of the theory being developed in the sequel.
One may notice that:
As for the more abundant symbol, here $a_1$, there are 14 instances with $\Delta_j=2$ (i.e., they directly precede instances of $a_2$ followed by one other instance of $a_1$) and 4 instances with $\Delta_j=1$ (i.e., they directly precede some other instance of $a_1$);
As for the less abundant symbol, here $a_2$, there are 4 instances with $\Delta_j=3$ and 10 instances with $\Delta_j=2$, the former having exactly one distance step more than the latter because of the presence of exactly 4 instances of $B_1=A_0=a_1$ symbol (see Figure \ref{fig:ruota}, second annulus from the outside), which increase exactly by one the distance between two consecutive instances of $B_0=a_2$ symbol (see Figure \ref{fig:ruota}, first annulus from the outside);
Moreover, said 4 instances of $B_1=A_0=a_1$ symbol are due to the repetitions of $A_3$ symbol in $C$, which are in the number of $D_3=\gcd(18,14)=2$, multiplied by the repetitions of $A_2$ symbol in $A_3$, which are in the number of $Q_3=\left\lfloor{4/2}\right\rfloor=2$.

\begin{algorithm}[H]\label{alg:esa}
 \vspace{10pt}
 \KwData{$(\mathcal{A},\mathbf{m})$}
 \KwResult{$C\in\Omega(\mathcal{A},\mathbf{m})$ such that $\sigma^2(C)=\min$; see Section \S\ref{sec:optimality}}
 \vspace{10pt}
 // initialization\\
 $A_0=a_{\arg\max\{m_1,m_2\}}$; $P_1=|\mathcal{J}_{\arg\max\{m_1,m_2\}}|$\;
 $B_0=a_{\arg\min\{m_1,m_2\}}$; $D_1=|\mathcal{J}_{\arg\min\{m_1,m_2\}}|$\;
 $i=0$; $R_i=1$\;
 \vspace{10pt}
 // looping is iterated until a null rest is found\\
 \While{$R_i>0$}{
  $i=i+1$\;
  $N_i=P_i+D_i$\;
  $Q_i=\left\lfloor{P_i/D_i}\right\rfloor$\;
  $R_i=P_i-Q_iD_i$\;
  $A_{i}=A_{i-1}^{Q_{i}}B_{i-1}$\;
  $B_{i}=A_{i-1}$\;
  $P_{i+1}=D_i$\;
  $D_{i+1}=R_i$\;
 }
 \vspace{10pt}
 // finalization\\
 $C=A_i^{D_i}$;
 \vspace{10pt}
 \caption{The Euclidean Sequencing Algorithm (ESA).}
\end{algorithm}
\vspace{10pt}

In general terms, the ESA may be described as follows; see Algorithm \ref{alg:esa}, where a pseudo-code implementation is sketched.
As for the initialization, $A_0$ denotes the symbol $a_k$, $k\in\{1,2\}$, with the highest multiplicity, and $B_0$ denotes the other symbol.
(If $m_1=m_2$, $A_0$ and $B_0$ may be used in any of the two possibilities, without prejudicing the algorithm.)
Let $P_1$ (resp., $D_1$) be the multiplicity of the symbol denoted by $A_0$ (resp., $B_0$).

The algorithm then proceeds iteratively.
At the $i$th iteration, a cyclic sequencing problem $(\mathcal{A}_i,\mathbf{m}_i)$ is adressed, where $\mathcal{A}_i:=[A_{i-1},B_{i-1}]$ is the alphabet and $\mathbf{m}_{i}:=[P_i,D_i]$ is the multiplicity vector.
Let $N_i:=P_i+D_i$ be the total number of items.
The quotient and the rest of the integer division $P_i/D_i$ are denoted by $Q_i:=\left\lfloor{P_i/D_i}\right\rfloor$ and $R_i:=P_i-Q_iD_i$, respectively.
The decomposition $P_i=Q_iD_i+R_i$, with $0\le R_i<D_i$, always exists and is unique.
If $R_i>0$, the new symbols $A_{i}:=A_{i-1}^{Q_{i}}B_{i-1}$ and $B_{i}:=A_{i-1}$ are defined, with multiplicities $P_{i+1}=D_i$ and $D_{i+1}=R_i$, respectively, and a new iteration is run.

Otherwise, i.e., when $R_i=0$, the stopping criterion is reached and the final cycle $C\in\Omega(\mathcal{A}_i,\mathbf{m}_i)$ is obtained as $C:=A_i^{D_i}=A_i^{\gcd(m_1,m_2)}$, where $D_i:=\gcd(m_1,m_2)$ (at the last iteration only) is deduced by comparison with the Euclidean algorithm for the gcd of two integers.

\begin{remark}
Since $R_i<D_i$, at any iteration the algorithm always attributes $A_i$ to the symbol with more instances, because
\begin{equation}\label{eqn:ESA:obviety}
D_i=R_{i-1}<D_{i-1}=P_i.
\end{equation}
\end{remark}

\begin{remark}
As intuitive, the cyclic sequencing problem size decreases iteration wise (stagnation is not possible) and the algorithm terminates in a finite number of steps.
In fact, (\ref{eqn:ESA:obviety}) and $R_i<D_i$ yield the strict inequality
\begin{equation}
|\mathbf{m}_{i+1}|=N_{i+1}=P_{i+1}+D_{i+1}=D_i+R_i<P_i+D_i=N_i=|\mathbf{m}_i|.
\end{equation}
\end{remark}

\begin{remark}\label{rem:gcd}
It is well-known from the theory of the classical Euclidean algorithm \cite{Euclid}, and easily proved, that $\gcd(m_1,m_2)=\gcd(P_i,D_i)$, $i\in\{1,\ldots,s\}$, where $s$ is the number of steps required to reach the end.
\end{remark}

\begin{definition}
Some instances of a same symbol $a_k\in\mathcal{A}$ are said to be unclustered in a cycle $C\in\Omega(\mathcal{A},\mathbf{m})$ if they are well separated in $C$, i.e., if $\Delta_j>1$ for all such instances.
\end{definition}

\begin{remark}\label{rem:unclustered}
Clearly, if some symbol $a\in\mathcal{A}$ is unclustered in the cycles $C_1,\ldots,C_r$, then $a$ is unclustered also in any product cycle $\prod_{k=1}^r C_k^{p_k}=C_1^{p_1}\cdot\ldots\cdot C_r^{p_r}$, for $p_k\in\mathbb{N}_0$, $k\in\{1,\ldots,r\}$, including the case of $C_1C_2$.
\end{remark}

\begin{lemma}\label{th:unclustered}
The Euclidean Sequencing Algorithm sets unclustered all instances of the less abundant symbol.
\end{lemma}

\begin{proof}
Induction on the number $s$ of steps required by the algorithm to end.
As for the base of the induction (i.e., $s=1$, standing for a case with perfect divisibility),
$$
C=A_1^{\gcd(m_1,m_2)}=\underbrace{A_1\ldots A_1}_{\gcd(m_1,m_2)}
=
\underbrace{A_0\ldots A_0}_{Q_1\geq 1}B_0
\underbrace{A_0\ldots A_0}_{Q_1\geq 1}B_0
\ldots
\underbrace{A_0\ldots A_0}_{Q_1\geq 1}B_0,
$$
so that the less abundant symbol, $B_0$, is always unclustered in $C$.
To prove the induction step, we consider a case requiring $s+1$ iterations and notice that, similarly,
$$
C=\underbrace{A_{s+1}\ldots A_{s+1}}_{\gcd(m_1,m_2)}
=
\underbrace{A_s\ldots A_s}_{Q_{s+1}\geq 1}B_s
\underbrace{A_s\ldots A_s}_{Q_{s+1}\geq 1}B_s
\ldots
\underbrace{A_s\ldots A_s}_{Q_{s+1}\geq 1}B_s.
$$

We now want to build a first auxiliary cyclic sequencing problem $(\mathcal{A}',\mathbf{m}')$ characterized in that:
i) $s$ steps are required to end, i.e., $R_s=0$;
ii) For $i\in\{1,\ldots,s\}$, the very same $Q_i$ (and then the very same $A_i$ and $B_i$) are produced as in $(\mathcal{A},\mathbf{m})$;
iii) $R_{s-1}=Q_{s+1}$, so that the cycle $C'=A_s^{D_s}=A_s^{R_{s-1}}=A_s^{Q_{s+1}}$ is produced by the ESA.
As noticed in Remark \ref{rem:gcd}, the latter requirement implies that $Q_{s+1}=\gcd(m'_1,m'_2)=\gcd(P_s,D_s)$, where $m'_1$ and $m'_2$ are the components of $\mathbf{m}'$.
Let us then set $D_s=x$, so that $P_s=Q_sD_s+R_s=Q_sx$ and $Q_{s+1}=\gcd(Q_sx,x)=x$.
Consequently, the last iteration (i.e., $i=s$) is fully determined.
Backtracking, one then computes $D_{i-1}=P_i$, $R_{i-1}=D_i$, $P_{i-1}=Q_{i-1}D_{i-1}+R_{i-1}$, for $i$ from $s$ back to 2, respecting the required constraints and thus producing the sought for auxiliary problem $(\mathcal{A}',\mathbf{m}')$.
All of the instances of $B_0$ are unclustered in $C'$ by the inductive hypothesis.

Similarly, we now want to build a second auxiliary cyclic sequencing problem $(\mathcal{A}'',\mathbf{m}'')$ characterized in that:
i) $s-1$ steps are required to end, i.e., $R_{s-1}=0$;
ii) For $i\in\{1,\ldots,s-1\}$, the very same $Q_i$ (and then the very same $A_i$ and $B_i$) are produced as in $(\mathcal{A},\mathbf{m})$;
iii) $R_{s-2}=1$, so that the cycle $C''=A_{s-1}^{D_{s-1}}=A_{s-1}^{R_{s-2}}=A_{s-1}=B_s$ is produced by the ESA.
As noticed in Remark \ref{rem:gcd}, the latter requirement implies that $1=\gcd(m''_1,m''_2)=\gcd(P_{s-1},D_{s-1})$, where $m''_1$ and $m''_2$ are the components of $\mathbf{m}''$.
Let us then set $D_{s-1}=x$, so that $P_{s-1}=Q_{s-1}D_{s-1}+R_{s-1}=Q_{s-1}x$ and $1=\gcd(Q_{s-1}x,x)=x$.
Consequently, the last iteration (i.e., $i=s-1$) is fully determined.
Backtracking, one then computes $D_{i-1}=P_i$, $R_{i-1}=D_i$, $P_{i-1}=Q_{i-1}D_{i-1}+R_{i-1}$, for $i$ from $s-1$ back to 2, respecting the required constraints and thus producing the sought for auxiliary problem $(\mathcal{A}'',\mathbf{m}'')$.
All of the instances of $B_0$ are unclustered in $C''$ by the inductive hypothesis.

As noticed in Remark \ref{rem:unclustered}, all of the instances of $B_0$ are unclustered in $(C'C'')^{\gcd(m_1,m_2)}=(A_s^{Q_{s+1}}B_s)^{\gcd(m_1,m_2)}=
A_{s+1}^{\gcd(m_1,m_2)}=C$.
\end{proof}

\begin{corollary}\label{th:more}
Let $(\mathcal{A},\mathbf{m})$ be a cyclic sequencing problem with $n=2$.
Then the Euclidean Sequencing Algorithm arranges the \textbf{more abundant symbol} in such a way that $|m_1-m_2|$ instances have $\Delta_j=1$, while the remaining have $\Delta_j=2$.
\end{corollary}

\begin{proof}
Let $a_{k'}$ be the less abundant symbol and $a_{k''}$ the more abundant.
Owing to Lemma \ref{th:unclustered}, the $m_{k'}$ instances of $a_{k'}$ are unclustered.
Therefore, there is the same number $m_{k'}$ of $a_{k''}$ symbol instances which are exactly before an instance of $a_{k'}$ symbol.
Such $a_{k''}$ symbol instances have $\Delta_j=2$.
All others $a_{k''}$ symbol instances are clustered, and thus they have $\Delta_j=1$, and there are $m_{k''}-m_{k'}=|m_1-m_2|$ such instances.
\end{proof}

\begin{definition}\label{def:approx}
Let $(\mathcal{A},\mathbf{m})$ be a cyclic sequencing problem. If $m_k\nmid N$, let
\begin{equation}
\ell_k:=\lfloor N/m_k\rfloor,\qquad
u_k:=\lceil N/m_k\rceil,
\end{equation}
be the floor and ceiling, resp., of the ratio $N/m_k\notin\mathbb{N}$.
\end{definition}

\begin{lemma}\label{th:twoValuesOnly}
Let $C\in\Omega(\mathcal{A},\mathbf{m})$ be an admissible cycle for a cyclic sequencing problem $(\mathcal{A},\mathbf{m})$ and let $a_k\in\mathcal{A}$ be a symbol such that $m_k\nmid N$. If $C$ is such that, for symbol $a_k$, the distance $\Delta_j$ may only take the two values $\ell_k$ and $u_k$, then
\begin{equation}\label{def:numbers}
N_{\ell,k}:=m_k u_k-N,\qquad
N_{u,k}:=N-m_k\ell_k
\end{equation}
are the number of symbol $a_k$ instances with $\Delta_j=\ell_k$ and with $\Delta_j=u_k$, respectively.
\end{lemma}

\begin{proof}
Notice first that, under the $m_k\nmid N$ assumption, $u_k=\ell_k+1$.
In order to grant feasibility, the linear system
$$
\left\{
\begin{array}{lllll}
N_{\ell,k}        & + & N_{u,k}     & = & m_k \\
N_{\ell,k} \ell_k & + & N_{u,k} u_k & = & N   \\
\end{array}
\right.
$$
must hold, where the first equation counts the total instances of symbol $a_k$, while the second equation expresses the round trip Lemma \ref{th:roundtrip}.
The unique solution to the system reads $N_{\ell,k}=m_k u_k-N\in\mathbb{N}$ and $N_{u,k}=N-m_k\ell_k\in\mathbb{N}$, as prospected.
From their definition, since $N,m_k,\ell_k,u_k\in\mathbb{N}\subset\mathbb{Z}$ and since $\mathbb{Z}$ is a ring, then $N_{\ell,k}\in\mathbb{Z}$ and $N_{u,k}\in\mathbb{Z}$.
Then, from the definitions of floor and ceiling, $m_ku_k > N > m_k\ell_k$, so that $N_{\ell,k}>0$ and $N_{u,k}>0$, and therefore the latter are admissible as quantities of symbol $a_k$ instances.
\end{proof}

\begin{remark}
Notice that a solution in $\mathbb{N}$ is found to a system of linear equations with coefficients in $\mathbb{N}$, which, in general, is only expected to yield a solution in $\mathbb{Q}$, not necessarily in $\mathbb{Z}$, nor even in $\mathbb{N}$.
In other terms, the problem at hand is a system of linear Diophantine equations, whose solution is fortunately simple because the system matrix has determinant $u_k-\ell_k=1$. For necessary and sufficient conditions for solution existence of systems of Diophantine equations, see \cite{Lazebnik} and references therein.
\end{remark}

\begin{remark}
Lemma \ref{th:twoValuesOnly} implies the intuitive result that, if $m_k\nmid N$, then it is not possible to have just one distance $\Delta_j$ for all instances of symbol $a_k$: at least two different distance values are necessary.
\end{remark}

\begin{corollary}\label{th:less}
Let $(\mathcal{A},\mathbf{m})$ be a cyclic sequencing problem with $n=2$ and let $a_k\in\mathcal{A}$ be the \textbf{less abundant symbol}.
Then, if $m_k\mid N$, the Euclidean Sequencing Algorithm arranges the instances of symbol $a_k$ in such a way that they all take the distance value $\Delta_j=N/m_k$.
Otherwise, i.e., if $m_k\nmid N$, the Euclidean Sequencing Algorithm arranges symbol $a_k$ in such a way that $N_{\ell,k}=m_k u_k-N$ instances take the distance value $\Delta_j=\ell_k$, while $N_{u,k}=N-m_k\ell_k$ take the distance value $\Delta_j=u_k$.
\end{corollary}

\begin{proof}
Considering the notation used to describe the ESA, $B_0$ represents the less abundant symbol (here denoted by subscript $k$) and it only appears in $A_1=A_0^{Q_1}B_0$, occupying the last position.
Notice that $Q_1=\left\lfloor{P_1/D_1}\right\rfloor=\left\lfloor{(N-D_1)/D_1}\right\rfloor=\left\lfloor{N/m_k-1}\right\rfloor$ instances of $A_0$ precede $B_0$ in $A_1$.

In the $m_k\mid N$ case, $Q_1=N/m_k-1$ exactly, and the ESA terminates in one step.
Since only $A_1$ symbols appear in the first step, any instance of $B_0$ is followed by $N/m_k-1$ instances of $A_0$. Therefore, each $B_0$ takes distance value $\Delta_j=N/m_k\in\mathbb{N}$.

In the $m_k\nmid N$ case, in the first step of the ESA, $A_1$ may only be followed either by a $B_1$ symbol or by another $A_1$ symbol, since these are the only two symbols appearing.
In the latter possibility, the next instance of $B_0$ happens after $\left\lfloor{N/m_k-1}\right\rfloor$ instances of $A_0$, so that $B_0$ takes the distance value $\Delta_j=\left\lfloor{N/m_k-1}\right\rfloor+1=\left\lfloor{N/m_k}\right\rfloor=\ell_k\in\mathbb{N}$.
In the former possibility, another instance of $A_0$ interposes before the next instance of $B_0$ occurs, so that $B_0$ takes the distance value $\Delta_j=\ell_k+1=u_k\in\mathbb{N}$.
Lemma \ref{th:twoValuesOnly} yields the number of instances with the two possible distance values.
\end{proof}

\begin{remark}\label{divisible}
Considering the $n=2$ case, let $m_1\ge m_2$, with $m_2\mid N$.
Then, $N/m_2=(m_1+m_2)/m_2=1+m_1/m_2\in\mathbb{N}$, and $m_1/m_2\in\mathbb{N}$.
By the argument in the proof above, the cycle returned by the ESA is characterized by repeating $m_2$ times the $a_1^{m_1/m_2}a_2$ pattern, that is,
$$
C=(a_1^{m_1/m_2}a_2)^{m_2}
=\overbrace{\underbrace{(a_1\ldots a_1}_{m_1/m_2}a_2)\cdot\ldots\cdot(\underbrace{a_1\ldots a_1}_{m_1/m_2}a_2)}^{m_2}.
$$
Accordingly, $\Delta_j=m_1/m_2+1=N/m_2$, $\forall j\in\mathcal{J}_2$, i.e., for all instances of the less abundant symbol $a_2$.
\end{remark}

\section{Optimality for binary cycles}
\label{sec:optimality}

From the analysis developed in the previous section, the cycle returned by the Euclidean Sequencing Algorithm has been completely characterized.
We claim it to be optimal, meaning variance minimizing, for binary cyclic sequencing problems.
This means that the ESA provides a direct, algorithmic solution to the
mixed integer quadratic problem \eqref{F.Obj}-\eqref{F.ConD} in the $n=2$ case.
The reason for this unexpected simplicity ultimately lies in the possibility to decompose the problem symbol-wise, and in that the global minimum is attained when each contributing symbol attains its own local minimum.

\begin{lemma}\label{th:lb}
A valid lower bound $LB$ on the optimal solution cost to problem \eqref{F.Obj}-\eqref{F.ConD} can be computed as
\begin{equation}\label{eqn:lb}
LB=\sum_{k=1}^n LB_k,
\end{equation}
where
\begin{alignat}{3}
LB_k:=\min  & \sum_{i \in \hat{\mathcal{J}}_k} \Delta^2_i \label{LB.Obj}\\
      s.t.  & \sum_{i \in \hat{\mathcal{J}}_k} \Delta_i=N,  \label{LB.roundtrip}\\
            & \Delta_i \geq 0, & \forall i \in \hat{\mathcal{J}}_k, \label{LB.ConD}
\end{alignat}
and the optimal distances are computed as
\begin{equation}\label{eqn:optdistance}
\Delta_i=\frac{N}{m_k},\qquad\forall i \in \hat{\mathcal{J}}_k.
\end{equation}
\end{lemma}

\begin{proof}
Problem \eqref{eqn:lb} is obtained from \eqref{F.Obj}-\eqref{F.ConD} by keeping the objective function and by surrogating constraints \eqref{F.dka} and \eqref{F.dkb} over $i\in \hat{\mathcal{J}}_k$ for a given $k\in\{1,\ldots,n\}$.
Problem \eqref{eqn:lb} can now be decomposed into $n$ subproblems of the kind \eqref{LB.Obj}-\eqref{LB.ConD}.
The latter may be addressed by the Lagrange multipliers method \cite{Rockafellar}.
Let
$$
L:=\sum_{i\in \hat{\mathcal{J}}_k}\Delta^2_i+2\lambda\left(N-\sum_{i\in \hat{\mathcal{J}}_k}\Delta_i\right)
$$
be the Lagrangian, where the factor $2$ has been introduced in $2\lambda$ for the sake of convenience.
Optimality is characterized by nullity of all partials, including the one with reference to Lagrange multiplier $\lambda$, that is,
$$
\left\{
\begin{array}{l}
\displaystyle 0=\frac{\partial L}{\partial\Delta_i}=2(\Delta_i-\lambda),\qquad\forall i\in \hat{\mathcal{J}}_k \\
\\
\displaystyle 0=\frac{\partial L}{\partial\lambda}=N-\sum_{i\in \hat{\mathcal{J}}_k}\Delta_i. \\
\end{array}
\right.
$$
Since the goal function is positive definite, optimality coincides with minimality.
A unique solution \eqref{eqn:optdistance} is straightforwardly found and positive, since $N>0$ and $m_k>0$.
\end{proof}

\begin{remark}
A sample cycle that satisfies lower bound \eqref{LB.Obj} for symbol $a_k$ is
$$
C=****a_k****a_k****a_k****a_k****a_k****a_k****a_k****a_k,
$$
where $N=40$, $m_k=8$, while jolly code $*$ denotes any other symbol than $a_k$. Instances of $a_k$ do not cluster together, but rather, stay isolated from each other.
\end{remark}

\begin{remark}
As from Remark \ref{divisible}, in the $n=2$ case, $a_2$ being the less abundant symbol and if $m_2\mid N$, the ESA returns a cycle such that \eqref{eqn:optdistance}, which is integer, is satisfied and thus lower bound \eqref{LB.Obj} is attained for the less abundant symbol.
\end{remark}

\begin{remark}\label{rem:equi}
Let us consider the case where all symbols have the same multiplicity, i.e., $\overline{m}=m_k=N/n$, $\forall k\in\{1,\ldots,n\}$, and the perfect divisibility hypothesis is satisfied (i.e., $N/\overline{m}=n\in\mathbb{N}$).
A cycle of the form $C=(a_1\ldots a_n)^{\overline{m}}$ (or obtained by means of repeating $\overline{m}$ times any permutation of the alphabet)
is characterized by $\Delta_j=n=N/\overline{m}$ for all instances of all symbols, i.e., \eqref{eqn:optdistance} is always verified.
Therefore, lower bound \eqref{eqn:lb} is attained and $C$ is optimal.
In the special $n=2$ case, this intuitive cycle is also returned by the Euclidean Sequencing Algorithm in one step, because $P_1=D_1=\overline{m}$, $Q_1=P_1/D_1=1$, $R_1=0$, $A_1=A_0^{Q_1}B_0=A_0B_0$ and $C=A_1^{D_1}=A_1^{\overline{m}}=(A_0B_0)^{\overline{m}}=(a_1a_2)^{\overline{m}}$, the latter coinciding with $(a_2a_1)^{\overline{m}}$ up to a rotation of the cycle by one step.
\end{remark}

In the general case, $N/m_k$ needs not be integer.
It is natural to conjecture whether the closest integers to $N/m_k$ solve the minimality problem for symbol $a_k$. This will appear to be the case.

\begin{lemma}\label{th:notsoeasylemma}
Let $(\mathcal{A},\mathbf{m})$ be a cyclic sequencing problem with $n=2$.
If $m_k\nmid N$ for some symbol $a_k\in\mathcal{A}$, then let $\ell_k:=\lfloor N/m_k\rfloor$ and $u_k:=\lceil N/m_k\rceil$ be as in Definition \ref{def:approx}.
Then, there exists a cycle $C\in\Omega(\mathcal{A},\mathbf{m})$ with $N_{\ell,k}:=m_k u_k-N$ instances of symbol $a_k$ with $\Delta_j=\ell_k$, along with $N_{u,k}:=N-m_k\ell_k$ instances of symbol $a_k$ with $\Delta_j=u_k$. The cycle $C$ as above minimizes the $(2,k)$th sub-moment $\mathcal{M}_2(a_k)$ for symbol $a_k$.
\end{lemma}

\begin{proof}
By Corollary \ref{th:less}, the prospected cycle $C$ exists, is admissible, and is returned by the ESA.
We now prove optimality.
Let $\mathbf{x}\in\mathbb{R}_+^{m_k}$, with the $i$th entry $x_i=\Delta_i=N/m_k$ as in \eqref{eqn:optdistance}. Let $\mathbf{z}\in\mathbb{N}^{m_k}$ be the proposed solution to the original, discrete problem.
It should be clear that $\mathbf{x}$, $\mathbf{z}$, and later $\mathbf{y}$ and $\mathbf{w}$, are here used to denote vectors collecting values of $\Delta_i$, $i\in \hat{\mathcal{J}}_k$, and that $m_k=|\hat{\mathcal{J}}_k|$.
Let $\pi_0\subset\mathbb{R}^{m_k}$ be the hyperplane of equation $\sum_{i\in \hat{\mathcal{J}}_k}x_i=N$, that is, the locus of points respectful of the round trip linear constraint.
As said, $\mathbf{z}$ is admissible, that is, it respects the round trip linear constraint, so that $\mathbf{z}\in\pi_0$; see Lemma \ref{th:twoValuesOnly}.
The same Lemma also shows that points (like $\mathbf{z}$) with $N_{\ell,k}$ coordinates equal to $\ell_k$ and $N_{u,k}$ coordinates equal to $u_k$ are all and the only points in the integer lattice $\mathbb{Z}^{m_k}$ to lie on $\pi_0$.
Actually, since the problem exhibits symmetry in coordinate permutations, $\mathbf{z}$ is representative of any point with this number and kind of coordinates, regardless the order.
The question then arise as if some point $\mathbf{w}\in\mathbb{N}^{m_k}\cap\pi_0$ exists, such that its coordinates are not permutations of those of $\mathbf{z}$, and yielding a better value than $\mathbf{z}$ of $\mathcal{M}_2(a_k)$.

We may now look at $\mathcal{M}_2(a_k)$ as the Euclidean distance $\|\cdot\|_2$ from the origin, to be minimized over a suitable feasible region.
In the unconstrained case, the trivial solution is $\mathbf{0}$; in the relaxed form \eqref{eqn:lb} the solution is the projection $\mathbf{x}\in\pi_0$ of $\mathbf{0}$ over $\pi_0$. This last assertion can be verified directly from the proof of Lemma \ref{th:lb}.
As a matter of fact, if $f(\mathbf{x}):=\sum_{j\in\hat{\mathcal{J}}_k}x_j^2$ is the $C^\infty$ goal function and $g(\mathbf{x}):=N-\sum_{j\in\hat{\mathcal{J}}_k}x_j$ is the $C^\infty$ constraint function for $\pi_0$, then $\mathbf{x}$ has been seen to satisfy the necessary optimality condition
$$
0=\nabla L(\mathbf{x})=\nabla f(\mathbf{x})+2\lambda\nabla g(\mathbf{x}),
$$
which is the KKT necessary optimality condition \cite{Karush,KuhnTucker} for the program at hand.
It follows that the gradient $\nabla f(\mathbf{x})=2\mathbf{x}$ of the goal function in the optimal point $\mathbf{x}$ must be parallel to $\nabla g(\mathbf{x})$, the normal to plane $\pi_0$.
Therefore $\mathbf{x}$ is the projection of $\mathbf{0}$ onto $\pi_0$.
Since the vector from the origin to $\mathbf{x}$ is orthogonal to plane $\pi_0$, then any other point $\mathbf{y}$ on $\pi_0$ is more distant from $\mathbf{0}$ than $\mathbf{x}$ is, as the Pythagorean theorem yields
$$
\|\mathbf{y}\|_2^2=\|\mathbf{x}\|_2^2+\|\mathbf{y}-\mathbf{x}\|_2^2,
\quad\forall\mathbf{y}\in\pi_0,
$$
where $\mathbf{y}-\mathbf{x}$ is a vector parallel to $\pi_0$.

In the case of our current interest, that is, with the constraint that  points must lie on $\pi_0$, both vectors $\mathbf{z}-\mathbf{x}$ and $\mathbf{w}-\mathbf{x}$ are parallel to $\pi_0$ and thus orthogonal to vector $\mathbf{x}$. Therefore, we need to prove that $\|\mathbf{z}-\mathbf{x}\|_2<\|\mathbf{w}-\mathbf{x}\|_2$ or, equivalently, that $\|\mathbf{z}-\mathbf{x}\|_2^2<\|\mathbf{w}-\mathbf{x}\|_2^2$ for, then, using the Pythagorean theorem twice one gets
$$
\|\mathbf{z}\|_2^2
=\|\mathbf{x}\|_2^2+\|\mathbf{z}-\mathbf{x}\|_2^2
<\|\mathbf{x}\|_2^2+\|\mathbf{w}-\mathbf{x}\|_2^2
=\|\mathbf{w}\|_2^2,
$$
which is equivalent to the optimality of $\mathbf{z}$ and thus to the thesis.

\begin{table}[h!]
\centering
\begin{tabular}{llllcl}
\hline
status & sacrificial coord.   &   $0$     & $1$      & $\ldots$ & $h-1$    \\
\hline
\\
original & $u_k$    & $\ell_k$ & $\ell_k$ & $\ldots$ & $\ell_k$ \\
\\
after $T_{u_k,-1}^{\ell_k,+1}$ & $\ell_k$ & $u_k$    & $\ell_k$ & $\ldots$ & $\ell_k$ \\
\\
after $T_{\ell_k,-(h-1)}^{\ell_k,+1}\circ T_{u_k,-1}^{\ell_k,+1}$ & $\ell_k-(h-1)=u_k-h$ & $u_k$ & $u_k$ & $\ldots$ & $u_k$    \\
\\
\hline
\\
original & $\ell_k$    & $u_k$ & $u_k$ & $\ldots$ & $u_k$ \\
\\
after $T_{\ell_k,+1}^{u_k,-1}$ & $u_k$ & $\ell_k$    & $u_k$ & $\ldots$ & $u_k$ \\
\\
after $T_{u_k,-1}^{u_k,+(h-1)}\circ T_{\ell_k,+1}^{u_k,-1}$ & $u_k+(h-1)=\ell_k-h$ & $\ell_k$ & $\ell_k$ & $\ldots$ & $\ell_k$    \\
\\
\hline
\end{tabular}
\caption{Visual proof of (\ref{eqn:Tul}), top, and of (\ref{eqn:Tlu}), bottom.}
\label{tab:composition}
\end{table}

One now notices that $\ell_k$ and $u_k$ are the closest integers to the real coordinates $x_j=N/m_k$ of $\mathbf{x}$, so that any other integer coordinate increases the distance from $\mathbf{x}$.
Nonetheless, one may think to produce an integer point $\mathbf{w}$ with a greater number of $\ell_k$ or of $u_k$ coordinates than $\mathbf{z}$, depending on whether $N/m_k$ be closer to $\ell_k$ or $u_k$, respectively.
Since feasible points must lie on $\pi_0$, such operation involves promoting some quantity $h$ of coordinates to a shorter distance from $N/m_k$, at the price of sacrificing at least one coordinate to a greater distance from $N/m_k$, in such a way that $\sum_{j\in\hat{\mathcal{J}}_k}\Delta_j=N$ be preserved.

Let us introduce the transformation
\begin{equation}\label{eqn:Tl}
T_{\ell_k,-h}^{\ell_k,+1}\,:\,\mathbb{Z}^{m_k}\cap\pi_0\rightarrow\mathbb{Z}^{m_k}\cap\pi_0
\end{equation}sending $h$ coordinates equal to $\ell_k$ to be $\ell_k+1=u_k$, while 1 coordinate equal to $\ell_k$ is sent to $\ell_k-h$. In the notation above, the quantity of coordinates affected by either change is implicitly determined by the claim to remain on the feasible plane $\pi_0$. Owing to the combinatorial invariance under coordinate permutations, the actual identity of affected coordinates is pointless. The same notation allows denoting by
\begin{equation}\label{eqn:Tu}
T_{u_k,-1}^{u_k,+h}\,:\,\mathbb{Z}^{m_k}\cap\pi_0\rightarrow\mathbb{Z}^{m_k}\cap\pi_0
\end{equation}
a transformation that sends $h$ coordinates equal to $u_k$ to be $u_k-1=\ell_k$, while 1 coordinate equal to $u_k$ is sent to $u_k+h$.
It is immediately seen that $T_{u_k,-1}^{\ell_k,+1}=T_{\ell_k,+1}^{u_k,-1}$ simply swaps one $\ell_k$ with one $u_k$ coordinate, which is not of interest.
Moreover, sacrificing the other coordinate, in order to allow promotions, in essence reduces to (\ref{eqn:Tl}) and (\ref{eqn:Tu}), because
\begin{equation}\label{eqn:Tul}
T_{u_k,-h}^{\ell_k,+1}=T_{\ell_k,-(h-1)}^{\ell_k,+1}\circ T_{u_k,-1}^{\ell_k,+1},
\end{equation}
\begin{equation}\label{eqn:Tlu}
T_{u_k,-1}^{\ell_k,+h}=T_{u_k,-1}^{u_k,+(h-1)}\circ T_{\ell_k,+1}^{u_k,-1},
\end{equation}
as visually reported in Table \ref{tab:composition}.
Therefore (\ref{eqn:Tl}) and (\ref{eqn:Tu}) are the only cases of interest.

\begin{table}[h!]
\centering
\begin{tabular}{|l|ll|ll|}
\hline
& \multicolumn{2}{l|}{\textbf{from}} & \multicolumn{2}{l|}{\textbf{to}} \\
\# & $\Delta_j$ & contribution & $\Delta_j$ & contribution \\
\hline
& & & & \\
$1$ & $\ell_k$ & $\left(\frac{N}{m_k}-\ell_k\right)^2$ & $\ell_k-h$     & $\left(\frac{N}{m_k}-\ell_k\right)^2+2\left(\frac{N}{m_k}-\ell_k\right)h+h^2$ \\
& & & & \\
$h$ & $\ell_k$ & $\left(\frac{N}{m_k}-\ell_k\right)^2$ & $\ell_k+1$ & $\left(\frac{N}{m_k}-\ell_k\right)^2-2\left(\frac{N}{m_k}-\ell_k\right)+1$ \\
& & & & \\
\hline
& & & & \\
$\Sigma$ & & $(1+h)\left(\frac{N}{m_k}-\ell_k\right)^2$ &  & $(1+h)\left(\frac{N}{m_k}-\ell_k\right)^2+h^2+h$ \\
& & & & \\
\hline
\end{tabular}
\caption{Effects of $T_{\ell_k,-h}^{\ell_k,+1}$, that is, one instance moves from distance $\ell_k$ to $\ell_k-h$, allowing $h$ instances to move from distance $\ell_k$ to $\ell_k+1=u_k$; first rows report individual contributions and quantities, while the last reports their sum.}
\label{tab:l}
\end{table}
\begin{table}[h!]
\centering
\begin{tabular}{|l|ll|ll|}
\hline
& \multicolumn{2}{l|}{\textbf{from}} & \multicolumn{2}{l|}{\textbf{to}} \\
\# & $\Delta_j$ & contribution & $\Delta_j$ & contribution \\
\hline
& & & & \\
$1$ & $u_k$ & $\left(u_k-\frac{N}{m_k}\right)^2$ & $u_k+h$     & $\left(u_k-\frac{N}{m_k}\right)^2+2\left(u_k-\frac{N}{m_k}\right)h+h^2$ \\
& & & & \\
$h$ & $u_k$ & $\left(u_k-\frac{N}{m_k}\right)^2$ & $u_k-1$ & $\left(u_k-\frac{N}{m_k}\right)^2-2\left(u_k-\frac{N}{m_k}\right)+1$ \\
& & & & \\
\hline
& & & & \\
$\Sigma$ & & $(1+h)\left(u_k-\frac{N}{m_k}\right)^2$ &  & $(1+h)\left(u_k-\frac{N}{m_k}\right)^2+h^2+h$ \\
& & & & \\
\hline
\end{tabular}
\caption{Effects of $T_{u_k,-1}^{u_k,+h}$, that is, one instance moves from distance $u_k$ to $u_k+h$, allowing $h$ instances to move from distance $u_k$ to $u_k-1=\ell_k$; first rows report individual contributions and quantities, while the last reports their sum.}
\label{tab:u}
\end{table}

The two transformations are dealt with in Table \ref{tab:l} and in Table \ref{tab:u}, respectively, showing the incremental effects onto the square of the distance from $\mathbf{x}$. Following simple computations, it is seen that the cumulated effect of any such move is not beneficial, for a net change equal to $h^2+h>0$ is produced.
Since $\mathbf{w}$ is to be produced by $\mathbf{z}$ throughout the (possibly repeated) application of (\ref{eqn:Tl}) and/or (\ref{eqn:Tu}), it follows that $\|\mathbf{z}-\mathbf{x}\|_2^2<\|\mathbf{w}-\mathbf{x}\|_2^2$.
This completes the prove.
\end{proof}

We are finally in the condition to minimize $\mathcal{M}_2(C)$, or equivalently $\sigma^2(C)$.

\begin{theorem}\label{th:main}
Let $(\mathcal{A},\mathbf{m})$ be a cyclic sequencing problem with a binary symbol alphabet.
Then the cycle $C\in\Omega(\mathcal{A},\mathbf{m})$ returned by the Euclidean Sequencing Algorithm minimizes variance.
\end{theorem}

\begin{proof}
The case of equally abundant symbols is discussed in Remark \ref{rem:equi} and agrees with ESA solution.
We now address the $m_1\neq m_2$ case and show that the cycle $C$ returned by the ESA, which has been characterized by Corollary \ref{th:less}, satisfies the sufficient optimality conditions.

Let $a_k$ be the more abundant symbol.
Clearly, $m_k\nmid N$.
Precisely, $0<N-m_k<m_k$ and $N/m_k=1+(N-m_k)/m_k\in(1,2)$.
Following Definition \ref{def:approx}, $\ell_k=\lfloor N/m_k\rfloor=1$ and $u_k=\lceil N/m_k\rceil=2$.
Considering (\ref{def:numbers}), since $\ell_k=1$, then $N_{u,k}=N-m_k\ell_k=N-m_k$ and, since $u_k=2$, then
$N_{\ell,k}=m_k u_k-N=2m_k-N=m_k-N_{u,k}$.
We conclude with Corollary \ref{th:more}, that the cycle $C$ returned by the ESA satisfies the conditions of Lemma \ref{th:notsoeasylemma}, thus minimizing the $(2,k)$th sub-moment $\mathcal{M}_2(a_k)$ for the more abundant symbol.

Let $a_k$ be the less abundant symbol. Owing to Corollary \ref{th:less}, the conditions of Lemma \ref{th:lb} if $m_k\mid N$ and of Lemma \ref{th:notsoeasylemma} if $m_k\nmid N$ are satisfied by the cycle $C$ returned by the ESA, which thus minimizes the $(2,k)$th sub-moment $\mathcal{M}_2(a_k)$ for the less abundant symbol.

It follows that ESA cycle $C$ minimizes the 2nd moment $\mathcal{M}_2(C)$ and, for Corollary \ref{th:optidentity}, it also minimizes the variance $\sigma^2(C)$.
\end{proof}

Theorem \ref{th:main} essentially collects results reached above.
The following, alternative and more direct finalization helps understanding that the simple reason behind ESA optimality is that more-than-linear powers, such as squares, heavily penalize deviations from patterns characterized as much as possible by uniform distances.

\begin{proof}[Alternative proof]
Without loss of generality, let $A_0$ (resp., $B_0$) be the more (resp., less) abundant symbol.
The ESA returns a pattern like
\begin{equation}\label{eqn:pattern}
C=
(\overset{1}{A_0}\ldots\overset{1}{A_0}\cdot\overset{2}{A_0}B_0)
\cdot
(\overset{1}{A_0}\ldots\overset{1}{A_0}\cdot\overset{2}{A_0}B_0)
\cdot
\ldots
\cdot
(\overset{1}{A_0}\ldots\overset{1}{A_0}\cdot\overset{2}{A_0}B_0),
\end{equation}
where the distance from one $A_0$ to the next is reported on top of the symbol.
As already discussed in the previous proof, Lemma \ref{th:notsoeasylemma} and Lemma \ref{th:lb} show that the ESA returned pattern (\ref{eqn:pattern}) minimizes $\mathcal{M}_2(B_0)$.
(In the spirit of directness proper of this proof, one may also easily show that the ESA pattern for $B_0$ minimizes the deviation from the average distribution.)
It remains to proof that pattern (\ref{eqn:pattern}) be optimal also for $A_0$, for the same argument as in the preceding proof leads to variance minimization.

Let us proceed by contradiction.
Any reshuffle of the pattern that keeps all $B_0$'s unclustered does not modify the distances between $A_0$'s, while the $B_0$ pattern may lose optimality. Therefore this case is not improving global optimality.
If a different kind of reshuffle brings some $B_0$'s to cluster, then $B_0$ optimality is worsened, for what above. Still, it could be worthwhile, provided that $A_0$ pattern improves more.

Since we are trying to minimize $\mathcal{M}_2(A_0)$, and since those $A_0$'s with $\Delta_j=1$ are already at the theoretical minimum, we need to decrease from 2 to 1 the distance for some other $A_0$.
Consider the round trip linear constraint $\sum_{j\in C^{-1}(A_0)}\Delta_j=N$ for the $A_0$'s: if some distance decreases, some other increases, in such a way that the sum of all distances remains equal to $N$.
Since decrements are from 2 to 1, each one contributes to a +1 increment somewhere else.
Let us conventionally single out and group the $A_0$'s that are contributing to increase the contribution of a same other prescribed $A_0$.
(Clearly, any reshuffle reduces to a combination of cases like this.)
Then, let $\Delta_{i_1},\ldots,\Delta_{i_r}$ be the distances of $r$ $A_0$'s decreasing from 2 to 1 and $\Delta_{i_0}$ the distance of the same other $A_0$ increasing by $r$.
\begin{itemize}

\item If $r=0$, nothing happens.

\item If $r=1$, then one $A_0$ decreases from 2 to 1 and either one increases from 1 to 2 (null net gain), or one other $A_0$ increases from $2$ to $3$ (net worsening when squaring).

\item If $r>1$, the gain is $r(2^2-1^2)=3r$, while the loss is either $(1+r)^2-1^2=r^2+2r$, when an increase from $1$ to $1+r$ happens, or $(2+r)^2-2^2=r^2+4r$, when an increase from $2$ to $2+r$ happens.
For $r>1$ integer, in the first case $r^2+2r>r+2r=3r$ and in the second case $r^2+4r>4r>3r$.

\end{itemize}
It follows that the loss is never lower than the gain, contradicting the possibility to improve $\mathcal{M}_2(A_0)$.
\end{proof}

\begin{remark}\label{rem:necessity}
The alternative proof, as well as the proofs of Lemma \ref{th:lb} and of Lemma \ref{th:notsoeasylemma}, indirectly yield in addition that, in the $n=2$ case, a unique distance $N/m_k$ or the number of instances $N_{\ell,k}$ and $N_{u,k}$ along with the distances correspondingly associated (as applicable depending on the divisibility condition), for all of the symbols in the alphabet, constitute a not only sufficient, but also necessary optimality condition.
Particularly, the less abundant symbol is necessarily unclustered in an optimal cycle.
\end{remark}

\section{Applications}
\label{sec:apps}

In some applicative contexts of particular relevance, it may be opportune to look for a possibly cyclic, variance minimizing sequence $C$ of the symbols from some alphabet $\mathcal{A}$ and with prescribed multiplicities $\mathbf{m}$.

Examples include flexible productive systems, able to to manufacture a large wealth of product types (here identified with the symbols in the alphabet), and that are efficiently run provided that the input sequence (here identified with $C$) is such to uniformly spread different product codes in time.
The binary case is exactly solved by the Euclidean Sequencing Algorithm.

A typical and easily understandable example is constituted by a multi-forking (particularly, by-forking) productive system  characterized by an initial branch common to all product types and then followed by a plurality (particularly, a couple) of finalization branches, each one being product-type specific.
Assume that no relevant cycle time difference exists among the finalization branches, but that the initial common branch be faster, which is a reasonable design choice owing to system topology.
In order to feed finalization branches uniformly, the rate of parts reaching the end of the common initial branch must be regular and uniform in time for any product code, in such a way that all branches are kept active and no parts accumulate in the crucial multi-fucation (particularly, bi-furcating) point.
Clearly, the impact of buffers may help managing difficult periods, but in the long run an efficient functioning is only possible under the respect of average supply rates from the upstream portion to downstream.
Equivalently, this demands the minimization of an index of dispersion around the mean value, exactly like variance $\sigma^2$ defined in (\ref{eqn:variance}).

The similitude between some aspects of industrial manufacturing systems and of electronic processing of information is well known. The conclusions above can thus be ported to the latter context.
Particularly, the problem of task scheduling in electronic processors provide interesting applications. Consider for instance a serial phase characterized by two distinct kinds of tasks, followed by two task-kind specific, parallel phases. Also in this case a uniform task distribution may result in improved performance. Applications to telecommunications are also easily envisaged, and particularly routing information packets throughout two distinct channels with the necessary elaboration performed by some suitable, common, upstream processing unit.

\begin{figure}[t!]
\centering
\includegraphics[width=0.8\linewidth]{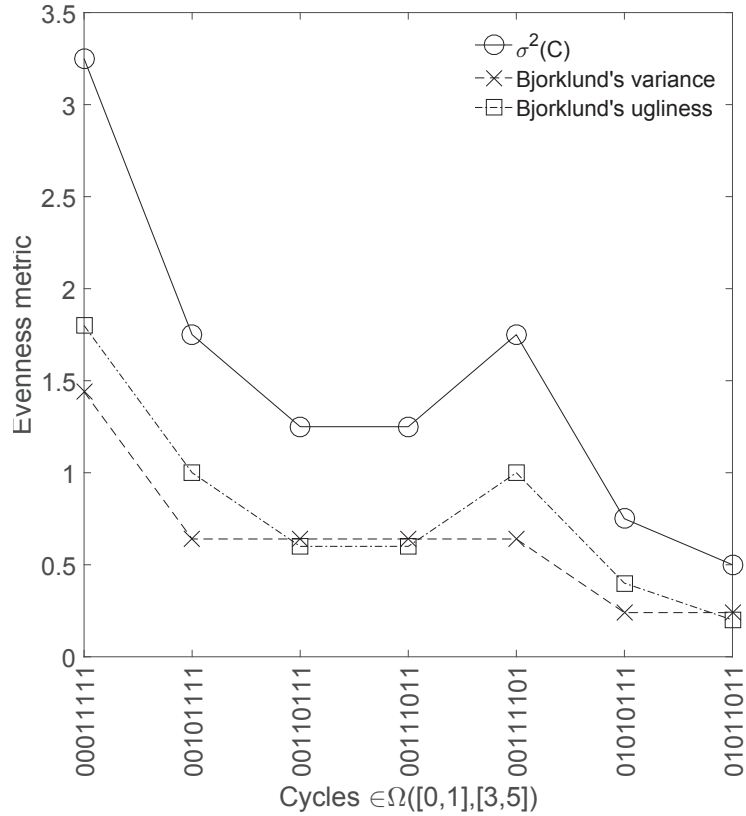}
\caption{Comparison between evenness metrics for $(\mathcal{A},\mathbf{m})=([0,1],[3,5])$: Proposed $\sigma^2(C)$ as in (\ref{eqn:variance}), solid; Bjorklund's variance, dashed, and ugliness, dash-dotted, as in \cite{Bjorklund2}; if 4-step rotations are applied, the last cycle on the right corresponds to ESA $01101101$, while the penultimate cycle on the right corresponds to non-ESA $01110101$.}
\label{fig:comparison}
\end{figure}

In both cases, manufacturing systems and information elaboration, the cyclic nature of the processes leads to the concept of cycles as a characteristic condition of interest. A similar cyclic request is found in applications to musical rhythms, as in Toussaint \cite{Toussaint}, and in applications to even pulse pattern generation in SNS accelerators in nuclear physics, as in Bjorklund \cite{Bjorklund1,Bjorklund2}.
As anticipated, in \cite{Bjorklund2} \emph{pulse sets} are only considered when computing evenness metrics.
Translated into the terminology of the present work, this means that, in Definitions \ref{def:mean} and \ref{def:variance}, summations over all items (i.e., $\sum_{j=1}^n$) have to be replaced by summations over pulse instances only (i.e., $\sum_{j\in\mathcal{J}_k}$, for the $k$ corresponding to symbol 1). According to Bjorklund and reformulating his statements in our notation, variance is unfit as evenness metric because, e.g., considering the cyclic sequencing problem $([0,1],[3,5])$, the cycles $C_1:=[0,1,1,0,1,1,0,1]$ (computed with the ESA) and $C_2=[0,1,1,1,0,1,0,1]$ (not computed with the ESA) have both Bjorklund-variance $\sigma^2=0.24$, $C_1$ being better looking than $C_2$ as for evenness.
According to the theory presently developed, $C_2$ is immediately seen not to satisfy the necessary (see Remarks \ref{rem:necessity}) and sufficient (see Lemma \ref{th:notsoeasylemma}) optimality conditions while $C_1$ does. Indeed, $\sigma^2(C_1)=0.5$ while $\sigma^2(C_2)=0.75$ with (\ref{eqn:variance}), so that, as expected from the theory, the case is correctly dealt with and the even pattern is preferred to the other.
When compared with Bjorklund's ugliness function \cite{Bjorklund2}, $\sigma^2(C)$ proves much simpler, both from the theoretical and the practical standpoint; see Figure \ref{fig:comparison}, where a comparison is shown for $(\mathcal{A},\mathbf{m})=([0,1],[3,5])$ between variance as proposed in the present work as opposed to Bjorklund's, along with Bjorklund's ugliness function.






\end{document}